\documentclass[aps, onecolumn, nofootinbib, footnote, superscriptaddress, longbibliography]{revtex4-2}

\usepackage{setspace} 
\usepackage[a4paper,left=2.2cm,right=2.2cm,top=3cm,bottom=3cm]{geometry}
\usepackage{soul}
\usepackage{relsize}
\usepackage{lmodern}
\usepackage{slantsc}
\usepackage{bbm}
\usepackage{graphicx}
\usepackage{amsmath}
\usepackage{amssymb}
\usepackage{amsthm}
\usepackage{amsbsy}
\usepackage{mathrsfs}
\usepackage{varioref}
\usepackage{dsfont}
\usepackage{bm}
\usepackage{mathtools}
\usepackage[usenames,dvipsnames]{xcolor}
\definecolor{myblue}{rgb}{0.2,0.2,0.8}
\definecolor{myblack}{rgb}{0,0,0}
\definecolor{myurl}{rgb}{0.1,0.1,0.4}
\usepackage[colorlinks=true,citecolor=myblue,linkcolor=orange,urlcolor=myurl]{hyperref}
\usepackage{cleveref}
\usepackage[font = footnotesize, labelfont = bf, justification = RaggedRight]{caption}
\usepackage{subfigure}
\usepackage{booktabs} 
\usepackage{array} 
\usepackage{paralist} 
\usepackage{verbatim} 
\edef\restoreparindent{\parindent=\the\parindent\relax}
\usepackage{multirow}

\usepackage[parfill]{parskip}
\setul{}{20 pt}

\newtheorem{definition}{Definition}
\newcommand{\defref}[1]{Definition~\ref{#1}}
\newtheorem{lemma}{Lemma}
\numberwithin{lemma}{section}
\newcommand{\lemref}[1]{Lemma~\ref{#1}}

\numberwithin{prop}{section}

\newtheorem{corollary}{Corollary}
\numberwithin{corollary}{section}
\newcommand{\corref}[1]{Corollary~\ref{#1}}
\newtheorem{remark}{Remark}
\numberwithin{remark}{section}
\newcommand{\remref}[1]{Remark~\ref{#1}}

\numberwithin{example}{section}

\newtheorem{theorem}{Theorem}
\numberwithin{theorem}{section}
\newcommand{\thmref}[1]{Theorem~\ref{#1}}

\numberwithin{assumption}{section}

\usepackage{pifont}

\newcommand{\<}{\langle}
\renewcommand{\>}{\rangle}

\newcommand{\rr}{{\mathcal{R}}}
\newcommand{\lo}{{\mathcal{L}}}

\newcommand{\cc}{{\mathcal{C}}}

\newcommand{\s}{{\mathcal{S}}}
\newcommand{\ee}{{\mathcal{E}}}
\newcommand{\xx}{{\mathcal{X}}}

\renewcommand{\aa}{{\mathcal{A}}}
\newcommand{\bb}{{\mathcal{B}}}

\newcommand{\mm}{{\mathcal{M}}}

\newcommand{\ii}{{\mathcal{I}}}
\newcommand{\jj}{{\mathcal{J}}}
\newcommand{\co}{\mathds{C}}
\newcommand{\re}{\mathds{R}}

\newcommand{\h}{{\mathcal{H}}}
\newcommand{\kk}{{\mathcal{K}}}
\newcommand{\hs}{{\mathcal{H}\sub{\s}}}
\newcommand{\ha}{{\mathcal{H}\sub{\aa}}}

\newcommand{\hr}{{\mathcal{H}\sub{\rr}}}
\newcommand{\hsys}{{H\sub{\s}}}

\newcommand{\E}{\mathsf{E}}

\newcommand{\Z}{\mathsf{Z}}

\newcommand{\one}{\mathds{1}}
\newcommand{\onesys}{\mathds{1}\sub{\s}}
\newcommand{\oneapp}{\mathds{1}\sub{\aa}}
\newcommand{\idch}{\operatorname{id}}
\newcommand{\idchsys}{\operatorname{id}\sub{\s}}

\newcommand{\zero}{\mathds{O}}

\newcommand{\tr}{\mathrm{tr}}
\newcommand{\tra}{\mathrm{tr}\sub{\aa}}
\newcommand{\trs}{\mathrm{tr}\sub{\s}}

\newcommand{\GO}{\operatorname{GLO}}

\newcommand{\sub}[1]{_{\!\mathsmaller{\, #1}}}

\newcommand{\eq}[1]{Eq.~\eqref{#1}}
\newcommand{\fig}[1]{Fig.~\ref{#1}}

\newcommand{\ket}[1]{|{#1}\rangle}

\newcommand{\rank}[1]{\mathrm{rank}\left( {#1}\right)}
\newcommand{\proj}[1]{\ensuremath{\left|#1\right\rangle\!\!\left\langle#1\right|}}

\makeatletter
\newcommand*\bigcdot{\mathpalette\bigcdot@{.5}}
\newcommand*\bigcdot@[2]{\mathbin{\vcenter{\hbox{\scalebox{#2}{$\m@th#1\bullet$}}}}}
\makeatother

\begin{document}

\title{
Thermodynamic closure of quantum measurements\\ and the limits of the indirect measurement model
}

\author{M. Hamed Mohammady}
\email{m.hamed.mohammady@savba.sk}
\affiliation{RCQI, Institute of Physics, Slovak Academy of Sciences, D\'ubravsk\'a cesta 9, Bratislava 84511, Slovakia}
\author{Francesco Buscemi}
\email{buscemi@nagoya-u.jp}
\affiliation{Graduate School of Informatics, Nagoya University, Furo-cho, Chikusa-Ku, Nagoya 464-8601, Japan}

\begin{abstract}
We investigate the consequences of requiring that a quantum measurement admit an adiabatic enclosure, so that it can be assigned a genuine thermodynamic description in which energy exchange is meaningfully resolved into work and heat. We identify two inequivalent levels at which such thermodynamic closure can be imposed: at the level of the measurement instrument acting on the system, or at the level of the indirect measurement process realising the instrument, namely the interaction between the system and a measuring apparatus, possibly including arbitrarily large environments. Although every instrument admits a unitary dilation, we show that such a construction does not generally provide an admissible thermodynamic closure.

This distinction is especially dramatic for \textit{efficient measurements}, i.e., measurements that are completely purity-preserving and described by a single Kraus operator, including von Neumann--L\"uders and square-root state-update rules. While efficient measurements are compatible with the laws of thermodynamics when thermodynamic closure is imposed at the level of the instrument, they are categorically forbidden when closure is imposed at the level of the measurement process. Our results therefore reveal a fundamental tension between thermodynamics and the universal applicability of the unitary interaction-based indirect measurement model.

\end{abstract}

\maketitle

\section{Introduction: measurements as thermodynamic processes}

Maxwell's demon is often invoked to illustrate how measurement may appear to challenge the laws of thermodynamics. Here we take the opposite route. Rather than asking how measurement might evade thermodynamic constraints, we ask what follows if those constraints are taken to apply universally, including to measurement itself. In this work, we adopt the premise that a quantum measurement admits a genuine thermodynamic description only if it is \textit{thermodynamically closable}, namely, only if it can be embedded into a larger process enclosed by an adiabatic boundary. In the spirit of Callen, adiabatic walls provide the operational prerequisite for the measurability of energy, and hence for resolving total energy exchange into work and heat. As he emphasizes~\cite{callen1991thermodynamics},
\begin{quote}
``An essential prerequisite for the measurability of the energy is the
existence of walls that do not permit the transfer of energy in the form of
heat.''
\end{quote}
This is the sense in which we understand thermodynamic closure in the present work. Thus, the central question for us is not whether a thermodynamic boundary can be drawn, but \textit{where} to draw it\footnote{This also seems to be related to the famous \textit{measurement problem}, i.e. how a microscopic interaction manifests itself in a definite macroscopic result. We will not discuss this question here.}.


The difficulty is that thermodynamic closure is not fixed by the usual mathematical description of measurement. Note that the \textit{unitary interaction measurement model} based on the concept of unitary dilation~\cite{Von-Neumann-Foundations,Neu40,Stinespring1955,gordon-louisell,RieszNagyFunct,Ozawa1984}, which states that any measurement can always be seen as a unitary interaction with an apparatus (or probe) initialized in a pure state, followed by a von Neumann (sharp, projective) measurement on the apparatus, does not resolve the issue, since it provides a \textit{purification}---i.e., a sort of \textit{information-theoretic closure}---of the initial measurement, but at the cost of shifting the actual measurement down the line, from the system to the apparatus. Moreover, it is not clear how the concept of information-theoretic closure relates to the concept of \textit{thermodynamic} closure. So the question remains: is there a unitary interaction measurement model that can be considered adiabatic from a thermodynamic point of view? Here we face a dilemma: either thermodynamic closure can be imposed at some level of description, or one must give up formulating a thermodynamics for quantum measurements. 

To address this, we \textit{assume} that quantum measurements are thermodynamically closable and, under this assumption, investigate how thermodynamic laws, especially the second and the third, impose constraints on both the measurement and its physical implementation. We find that such thermodynamic constraints are especially problematic for \textit{purity-preserving} measurements, also known as \textit{quasicomplete} measurements~\cite{Ozawa1986}---those in which all pure initial states of the system are mapped to pure final states, for any measurement outcome. Among these, particularly relevant are \textit{completely} purity-preserving measurements, also known as \textit{efficient} measurements~\cite{Jacobs2009}; in this case, each measurement outcome corresponds to a single Kraus operator, so that they preserve purity even when performed locally on a pure entangled state. Efficient measurements, which contain as a special case von 
Neumann--L\"uders and square-root measurements, are known for their nice mathematical properties and are frequently assumed, either explicitly or implicitly, in several works in quantum information theory and quantum thermodynamics. Some textbooks even present efficient measurements as the most general measurement model, referred to as the \textit{measurement operator formalism}~\cite{Nielsen-Chuang}. Our results reveal significant problems with this assumption, \textit{especially} in a thermodynamic context, where efficient measurements may be fundamentally unattainable~\cite{Guryanova2018, Mohammady2022a}.

\subsection{Summary of the main results}

In this work, we focus on  \emph{adiabatic} implementations of quantum measurements---either as  instruments acting on the measured system or as indirect measurement processes realised via interactions of the system with a measurement apparatus---and examine whether such implementations are consistent with the second and third laws of thermodynamics. After introducing the necessary notation to describe quantum measurements in full generality, we establish the following results:

\begin{enumerate}

\item If thermodynamic closure is carried out at the level of the \textit{instrument}, so that during  measurement  the system being measured does not exchange heat with the external environment, then \thmref{thm:efficient-instrument-both-laws} shows that such an adiabatically implemented instrument is thermodynamically consistent if and only if the Shannon entropy of the outcome probability distribution exceeds the Groenewold-Lindblad-Ozawa information gain, and at least one operation of the instrument is strictly positive, i.e., it maps all full-rank states to full-rank states. As a direct consequence, an observable admits an efficient measurement if and only if at least one effect in its range is strictly positive. That is, while projective von Neumann--L\"uders measurements of (non-trivial) observables are categorically forbidden as fundamentally incompatible with thermodynamics, square-root (i.e., generalized L\"uders) measurements of some unsharp POVMs are allowed.

\item   If thermodynamic closure is carried out at the level of the \textit{measurement process}, so that no heat is exchanged with the external universe during the state preparation of the apparatus,  the premeasurement interaction between system and apparatus, and  the pointer objectification mechanism acting on the apparatus, then the measurement process is thermodynamically consistent if and only if the premeasurement interaction is a bistochastic channel---for example, a unitary channel---and the measuring apparatus is prepared in a full-rank state. \thmref{thm:bistochastic-third-law-purity-preserving-nogo} shows that for any non-trivial observable, a thermodynamically consistent adiabatically implemented measurement process cannot result in a quasicomplete measurement, let alone an efficient one.

\end{enumerate}

\section{Preliminaries}

\subsection{Operations and channels}

Here we  consider only systems with complex Hilbert spaces $\h$ of finite dimension. Let $\lo(\h)$ be the algebra of linear operators on $\h$, with the symbols $\one$ and $\zero$ denoting the unit and null operators in $\lo(\h)$, respectively; an operator $E \in \lo(\h)$ such that $\zero \leqslant E \leqslant \one$ is called an \textit{effect}.  An effect is called \textit{trivial} whenever it is proportional to $\one$, i.e., $E = \alpha \one$ for some $\alpha \in [0,1]$. A \textit{projection} is an effect satisfying $E^2 = E$. A \textit{state} on $\h$ is defined as a positive semidefinite operator of unit trace, and the space of states  on $\h$ is denoted as $\s(\h) \subset \lo(\h)$. The extremal elements of $\s(\h)$ are the pure states, which are  rank-1 projections. For any state $\rho\sub{\h \kk}$ on a bipartite system $\h \otimes \kk$, we define $\rho\sub{\h} \coloneq \tr\sub{\kk}[\rho\sub{\h \kk}]$ and $\rho\sub{\kk} \coloneq \tr\sub{\h}[\rho\sub{\h \kk}]$ the reduced states of $\rho\sub{\h\kk}$ in $\h$ and $\kk$, respectively.  A self-adjoint operator $A \in \lo(\h)$ is called positive definite, or \textit{strictly positive}, if $A >\zero$, i.e., if all the eigenvalues of $A$ are strictly positive. If $A$ is strictly positive then it has full rank in $\h$, i.e., $\rank{A} = \dim(\h)$.

A  linear map $\Phi : \lo(\h) \to \lo(\kk)$ is called an \textit{operation} if it is completely positive and trace non-increasing. When $\kk = \h$, we say that the operation acts on $\h$.   A trace preserving operation is called a \textit{channel}. Consider the pair of operations $\Phi_1: \lo(\h_1) \to \lo(\kk_1)$ and $\Phi_2 : \lo(\h_2) \to \lo(\kk_2)$. The parallel composition of these operations is $\Phi_1 \otimes \Phi_2 : \lo(\h_1 \otimes \h_2) \to \lo(\kk_1 \otimes \kk_2), A \otimes B \mapsto \Phi_1(A) \otimes \Phi_2(B)$.  If $\kk_1 = \h_2$, the sequential composition is  $\Phi_2 \circ \Phi_1 : \lo(\h_1) \to \lo(\kk_2), A \mapsto  \Phi_2[\Phi_1(A)]$. The \textit{identity channel}, which maps every operator to itself, is denoted as $\idch$. For each operation $\Phi : \lo(\h) \to \lo(\kk)$ there exists a unique dual map $\Phi^* : \lo(\kk) \to \lo(\h)$ defined by the trace duality  $\tr[\Phi^*(A) B] = \tr[A \Phi(B)]$ for all $A \in \lo(\kk),B \in \lo(\h)$. An operation is \textit{compatible} with a unique effect $E \in \lo(\h)$ via the relation $\Phi^*(\one\sub{\kk}) = E$. If $\Phi$ is a channel, then $\Phi^*$ is \textit{unital}, i.e., $\Phi$ is compatible with the trivial effect $\one\sub{\h}$.

An operation $\Phi$ is called \textit{purity-preserving} (or just \textit{pure}) if $\rho$ pure $\implies$ $\Phi(\rho)$ pure; it is \textit{completely} purity-preserving if $\idch\otimes\Phi$ is purity-preserving on any extension $\h'\otimes\h$ of $\h$. An operation $\Phi$ is called \textit{strictly positive} if $ A >\zero \implies \Phi(A) > \zero$.  An operation $\Phi$ acting on $\h$ is called a \textit{bistochastic channel} if it preserves both the trace and the unit; this is possible if and only if its dual $\Phi^*$ is also bistochastic, i.e., $\Phi^*$ is not only unital, but it also preserves the trace. It is easy to verify that bistochastic channels are rank non-decreasing, and hence strictly positive; see, e.g.~\cite{Buscemi2005}.

\subsection{Quantum instruments and measurement processes}

Let us consider a quantum system $\mathcal{S}$ associated with a finite-dimensional Hilbert space $\hs$. An observable on $\hs$ is represented by a normalized \textit{positive operator-valued measure} (POVM). We  consider only discrete observables, which are identified with the family $\E\coloneq  \{E_x : x \in \xx\}$, where $\xx = \{x_1, \dots, x_N\}$ is a (finite) alphabet (also called value space or the space of measurement outcomes) and  $E_x$ are effects in $\lo(\hs)$, normalized so that $\sum_{x\in \xx} E_x = \onesys$. The probability of observing outcome $x$ when measuring $\E$ in the state $\rho$ is given by the Born rule as $p^{\E}_\rho(x) \coloneq  \tr[E_x \rho]$. An observable is \textit{non-trivial} if at least one effect in its range is non-trivial, which implies that $|\xx| = N$ must be larger than one. An observable is  a \textit{projection valued measure} (PVM), or projective,   if $E_x$ are mutually orthogonal projections, i.e., $E_x E_y = \delta_{x,y} E_x$.  We restrict ourselves only to observables such that $E_x \ne \zero$ for all $x$: this is always possible by replacing the original value space $\xx$ with the relative complement $\xx \backslash \{x: E_x = \zero\}$.

The state-update rule for measurement of a (discrete) POVM is described by an (discrete) \textit{instrument} acting on $\hs$, which is given by a family of operations $\ii \coloneq  \{\ii_x : x\in \xx\}$ acting on $\hs$, normalized so that the average expectation $\ii_\xx(\bigcdot ) \coloneq  \sum_{x\in \xx} \ii_x(\bigcdot )$ is a channel. Each instrument is associated with a unique observable $\E$ via $\ii_x^*(\onesys) = E_x$. In this case, we say that the instrument $\ii$ is $\E$-\textit{compatible}, or that it is an $\E$-instrument for short. While every instrument is compatible with a unique observable, every observable admits infinitely many instruments; all $\E$-instruments can be written as $\ii_x(\bigcdot) \coloneq \Phi_x \circ \ii^L_x(\bigcdot)$ where $\ii^L_x(\bigcdot) \coloneq \sqrt{E_x} (\bigcdot) \sqrt{E_x}$ is the generalised L\"uders (or square-root) instrument for $\E$ and $\Phi_x$ are arbitrary channels that may depend on outcome $x$.  It is customary to represent explicitly also the classical register $\rr$ in which the measurement outcomes are stored. In order to unify the notation, the register, although intrinsically classical, is also represented by a Hilbert space $\h\sub{\rr}$, and the outcomes $x$ of  measurement are  stored in the register as the corresponding element of the orthogonal set of unit vectors $\{\ket{x} \in \h\sub{\rr}\}$. The ``augmented'' instrument channel $\Xi^\ii : \lo(\hs ) \to \lo(\hs \otimes \hr) $ is defined as 
\begin{align}\label{eq:register-instrument-channel}
    \Xi^\ii(\bigcdot\sub{\s} ) \coloneq \sum_{x\in \xx} \ii_x(\bigcdot\sub{\s}) \otimes \proj{x}\sub{\rr} \, .
\end{align}
This channel packages into a single quantum-classical output the average post-measurement state produced by the instrument: for each outcome $x$, the corresponding operation $\ii_x$ determines the conditional output state of the system, while the projector $\proj{x}\sub{\rr}$ records that outcome in the register. Thus, for any prior state $\rho\sub{\s}$, $\Xi^\ii(\rho\sub{\s})$ is precisely the expected joint state of system and register after the measurement, before conditioning on a particular outcome. Note that $\tr\sub{\rr} \circ \Xi^\ii(\bigcdot) = \ii_\xx(\bigcdot)$.

For any prior state $\rho\sub{\s}$, the expected posterior state of system and register after measurement by $\ii$ reads
\begin{align}\label{eq:avg-posterior-state}
    \Xi^\ii(\rho\sub{\s}) = \sum_{x\in \xx} p^\E_\rho(x)  \sigma\sub{\s}^x \otimes \proj{x}\sub{\rr}\, ,
\end{align}
where 
\begin{align}\label{eq:posterior-states}
\sigma\sub{\s}^x \coloneq  \frac{1}{p^{\E}_\rho(x)}\ii_x(\rho\sub{\s})\;,\qquad p^{\E}_\rho(x)  >0\;,
\end{align}
are the posterior  states of the system conditional on observing outcome $x$. In the above, if $p^{\E}_\rho(x)  =0$, the posterior  state can be defined arbitrarily.





\

\begin{definition}
An instrument $\ii\coloneq \{\ii_x : x\in \xx\}$  is called \emph{quasicomplete} if all operations $\ii_x$ are purity-preserving~\cite{Ozawa1986}. A subclass of quasicomplete instruments are called \emph{efficient}, i.e., instruments for which every operation is completely purity-preserving. This is equivalent to saying that every operation of an efficient instrument is expressible with a single Kraus operator, i.e.,  $\ii_x(\bigcdot ) = U_x \sqrt{E_x} \, \bigcdot \, \sqrt{E_x} U_x^*$, where $E_x$ are the effects compatible with each operation $\ii_x$ and $U_x$ are unitary operators.   
\end{definition}

\

While an instrument is sufficient for describing the state update of a system being measured, it does not take account of the role played by the measuring apparatus; a  \textit{measurement process} for an instrument is a specification of how the quantum system to be measured interacts with a given (quantum) measurement apparatus. Let  $\aa$  be an auxiliary system (the \textit{apparatus}) with Hilbert space $\ha$ which, since we only consider finite-dimensional systems $\hs$ in this paper, can always be taken finite dimensional; let $\xi$ be a state on $\ha$; let $\ee$ be a \textit{premeasurement channel} acting on $\hs\otimes \ha$; and let $\jj$ be an \textit{objectification} instrument acting on $\ha$ that is compatible with a \textit{pointer observable} $\Z\coloneq \{ Z_x : x \in \xx\}$, which can be a general POVM.   The tuple $\mm\coloneq  (\ha, \xi\sub{\aa}, \ee\sub{\s\aa}, \jj\sub{\aa})$ is a measurement process for a unique $\E$-compatible instrument $\ii$ acting on $\hs$, with the same value space $\xx$ as that of the pointer observable,   via the relations
\begin{align}
   & \ii_x(\bigcdot \sub{\s}) \coloneq  \Big[\tra\circ (\idchsys \otimes \jj_x) \circ \ee\sub{\s\aa}\Big](\bigcdot \sub{\s} \otimes \xi\sub{\aa}) \equiv \tra[(\onesys \otimes  Z_x)\ \ee\sub{\s\aa}(\bigcdot \sub{\s} \otimes \xi\sub{\aa})]\;,\nonumber\\
   &\ii_x^*(\bigcdot \sub{\s}) \coloneq \tra[\ee^*\sub{\s\aa}(\bigcdot \sub{\s} \otimes  Z_x)\ (\one\sub{\s} \otimes \xi\sub{\aa})] \equiv [\Gamma_\xi\circ \ee^*\sub{\s\aa}](\bigcdot \sub{\s} \otimes  Z_x)\; ,\label{eq:system-instrument}
\end{align}
for all $x \in \xx$. In the above,  $\Gamma_\xi : \lo(\hs\otimes \ha) \to \lo(\hs)$ is a conditional expectation with respect to $\xi\sub{\aa}$, also called a \textit{restriction map}, defined as
\begin{align}\label{eq:restriction-map}
\Gamma_\xi(\bigcdot \sub{\s\aa}) \coloneq  \tra[\bigcdot \sub{\s\aa}\ (\one\sub{\s} \otimes \xi\sub{\aa})]\;,
\end{align}
so that $\tr[\Gamma_\xi(\bigcdot \sub{\s\aa})\ \rho\sub{\s}] \coloneq  \tr[\bigcdot \sub{\s\aa}\ (\rho\sub{\s} \otimes \xi\sub{\aa})]$ for all states $\rho\sub{\s}$,  by construction. Such a map is obviously unital and completely positive.  Our definition generalizes that given in~\cite{Ozawa1984}, where $\ee$ is unitary,  $\xi$ is pure, and $\Z$ is projection valued. Since the instrument realised on the system  depends only on the pointer observable $\Z\sub{\aa}$ but not on the instrument $\jj\sub{\aa}$ used to measure it, see \eq{eq:system-instrument},  any apparatus instrument $\jj\sub{\aa}$ compatible with the same pointer observable $\Z\sub{\aa}$ realises the same instrument on the system, all else being equal. However, it is important to emphasize that different realisations $\jj\sub{\aa}$ of the same pointer observable $\Z\sub{\aa}$ may well have different physical or thermodynamic properties. This observation will be important in what follows.

Just as an instrument on the system can be represented by the augmented channel $\Xi^\ii$, a measurement process naturally defines an instrument on the joint system $\hs \otimes \ha$, with the apparatus now treated as part of the quantum output prior to conditioning on the observed outcome. Specifically, the measurement process $\mm$ induces the ``global'' instrument
\begin{align}\label{eq:global-instrument}
 (\idchsys \otimes \jj) \circ \ee \coloneq \{(\idchsys \otimes \jj_x) \circ \ee : x\in \xx\}  \, ,
\end{align}
acting on $\hs \otimes \ha$, which is compatible with the ``Heisenberg-evolved'' pointer observable $\Z^\tau \coloneq \{Z_x^\tau \equiv \ee^*(\onesys \otimes Z_x) : x \in \xx \}$. By \eq{eq:register-instrument-channel}, and in direct analogy with \eq{eq:avg-posterior-state}, the corresponding augmented channel describes the expected joint post-measurement state of system, apparatus, and register. Thus, for any prior state $\rho\sub{\s}$ of the system and initial apparatus state $\xi\sub{\aa}$, we have
\begin{align*}
  \Xi^{(\idchsys \otimes \jj) \circ \ee}(\rho\sub{\s} \otimes \xi\sub{\aa}) & = \sum_{x\in \xx} p^{\E}_\rho(x) \sigma\sub{\s\aa}^x \otimes |x\rangle\!\langle x|\sub{\rr}\;.
\end{align*}



\subsection{Information measures}

Recalling that $\mathscr{H}(\{p_x\})\coloneq -\sum_x p_x \ln p_x$ is the \textit{Shannon entropy} for the probability distribution $\{p_x\}$ and  $S(\rho) \coloneq  -\tr[\rho \ln(\rho)]$ is the von Neumann entropy for the quantum state $\rho$,  we define the following information quantities for an observable $\E$, measured by an instrument $\ii$, in the state $\rho$:
\begin{align}\label{eq:entropic-quantities}
\mathscr{H}(p^{\E}_\rho) &\coloneq  - \sum_{x\in \xx} p^{\E}_\rho(x) \ln p^{\E}_\rho(x) \geqslant 0\;, \nonumber \\
I_{\GO}(\ii, \rho) &\coloneq   S(\rho\sub{\s}) - \sum_{x\in \xx} p^{\E}_\rho(x) S(\sigma\sub{\s}^x) \gtreqqless 0  \;.
\end{align}
In the above,  $\mathscr{H}(p^{\E}_\rho)$ is the Shannon entropy of the probability distribution $p^{\E}_\rho$ obtained when measuring the observable $\E$ in the state $\rho$.  The quantity   $I_{\GO}(\ii, \rho)$ is the system's \textit{Groenewold--Lindblad--Ozawa (GLO) information gain}~\cite{Groenewold1971,Lindblad1972,Ozawa1986}: it is uniquely determined by the measurement of the $\E$-compatible instrument $\ii$  in the prior system state $\rho$, where we recall that $\sigma^x\sub{\s}$ are the posterior states defined in \eq{eq:posterior-states},  and is guaranteed to be non-negative for all states $\rho$ if and only if $\ii$ is quasicomplete \cite{Ozawa1986}.  Note that while $\mathscr{H}(p^{\E}_\rho)$    depends only on the observable $\E$ measured in the system, $I_{\GO}(\ii, \rho)$ also depends on the instrument that measures this observable.

\section{Consistency of adiabatic instruments and measurement processes with thermodynamics}

\begin{figure}[!htb]
    \centering
    \makebox[\textwidth][c]{
        \includegraphics[width=0.9\textwidth]{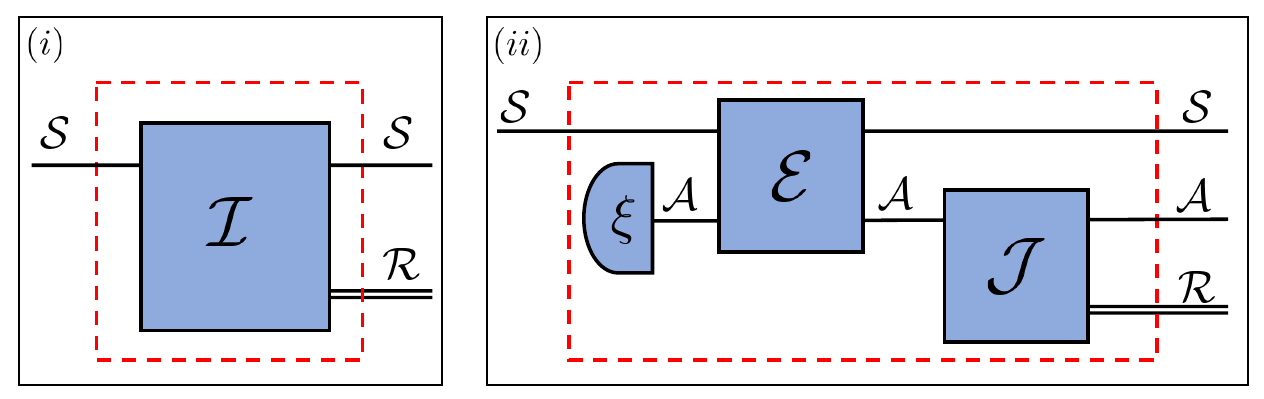}
    }
    \caption{Solid black lines are quantum systems---the system being measured $\s$ and measuring apparatus $\aa$---and the doubled black line is the classical register $\rr$.  The blue boxes are instruments ($\ii$ and $\jj$) that receive a quantum system from the left and produce a quantum system and classical information on the right, channels ($\ee$) that take quantum systems from the left to quantum systems on the right, and state preparations ($\xi$) producing a quantum system on the right in a fixed state. The red dashed box represents the adiabatic boundary, so that any  blue boxes contained within are implemented adiabatically. $(i)$ The adiabatic boundary is drawn over the instrument $\ii$ acting on the system and producing post-measurement states and classical measurement outcomes. $(ii)$ The adiabatic boundary is drawn over the measurement process implementing $\ii$, which includes state preparation of the apparatus $\xi$, premeasurement interaction between system and apparatus $\ee$, and objectification instrument $\jj$ acting on the apparatus.  }
    \label{fig:closure}
\end{figure}

As discussed in the introduction, in order to formulate a thermodynamics of quantum measurements, we assume that either (i) the instrument itself, or (ii) the measurement process realising the instrument, are implemented \emph{adiabatically}.  In the former case, we draw an ``adiabatic boundary" over the compound of measured system and classical register  during the implementation of the instrument, so that the compound is  \textit{thermally closed},  exchanging at most mechanical energy (work) with an external source, but not heat. In the latter case, we extend the adiabatic boundary to also include the  measuring apparatus---which may also include  (finite) subsystems playing the role of environments or heat baths, see, e.g.,  \cite{Strasberg2020b, Latune2024}---used to implement the instrument \cite{Minagawa2023a}. Here, the initial state preparation of the apparatus, the premeasurement interaction between system and apparatus, and the final objectification instrument on the apparatus, are all adiabatically implemented. See \fig{fig:closure} for a schematic representation for the two scenarios. 

As we shall see below, the level at which thermal closure is performed---the instrument or the measuring process---will have consequences for the consistency of measurements with the laws of thermodynamics.  

\subsection{Consistency with the second law}

The Clausius statement of the second law is that the work extracted by an isothermal process, i.e., a process  involving an exchange of heat with at most a single thermal bath, is bounded by the change in free energy resulting from the process. While this statement has been extended to the quantum regime and beyond equilibrium, it is usually formulated for processes where the input and output systems coincide.  We now generalise this to channels for arbitrary inputs and outputs:

\


\begin{definition}[Operational formulation of the second law]\label{def:second-law}
Let $\h$ and $\kk$ be systems with Hamiltonians $H\sub{\h}$ and $H\sub{\kk}$, respectively. A channel $\Lambda:\lo(\h)\to\lo(\kk)$ is said to admit an \emph{isothermal implementation} at inverse temperature $\beta$ if its implementation involves an exchange of heat with at most one thermal environment at inverse temperature $\beta$. The \emph{second law} is taken to assert that, for any such isothermal implementation and any initial state $\rho\in\s(\h)$, the extracted work is bounded by the decrease in non-equilibrium free energy:
\begin{align*}
W_\mathrm{ext}(\rho \mapsto \Lambda(\rho))
\leqslant
F(\rho,H\sub{\h},\beta)-F(\Lambda(\rho),H\sub{\kk},\beta)\, ,
\end{align*}
where $F(\sigma,H,\beta)\coloneq \tr[H\sigma]-\beta^{-1}S(\sigma)$
is the non-equilibrium free energy of a system with Hamiltonian $H$ in state $\sigma$, relative to inverse temperature $\beta$.
\end{definition}

In particular, an adiabatically implemented channel is a special case of an isothermally implemented one, namely, the limiting case in which no heat is exchanged with the thermal environment.

\
\begin{lemma}\label{lemma:adiabatic-channel-second-law}
    Let $\h$ and $\kk$ be systems with Hamiltonians $H\sub{\h}$ and $H\sub{\kk}$, respectively. Let $\Phi : \lo(\h) \to \lo(\kk)$ be an adiabatically implemented channel, i.e., such that its implementation does not result in an exchange of heat with any thermal environment, so that for any initial state $\rho \in \s(\h)$ the work extracted by $\Phi$ reads
    \begin{align*}
      W_\mathrm{ext}(\rho \mapsto \Phi(\rho)) = \tr[ H\sub{\h} \rho] - \tr[ H\sub{\kk} \Phi(\rho)] \, .   
    \end{align*}
    The channel $\Phi$ is consistent with the second law if and only if $S(\Phi(\rho)) \geqslant S(\rho)$ for all $\rho \in \s(\h)$. 
\end{lemma}
\begin{proof}
 For any $\beta$ and initial state $\rho$,  we may write 
 \begin{align*}
 W_\mathrm{ext}(\rho \mapsto \Phi(\rho)) = -\Delta F_{\h \to \kk} - \beta^{-1} \Delta S_{\h \to \kk} \, ,    
 \end{align*}
where $\Delta F_{\h \to \kk} \coloneq F(\Phi(\rho), H\sub{\kk}, \beta) - F(\rho, H\sub{\h}, \beta)$ and $\Delta S_{\h \to \kk} \coloneq S(\Phi(\rho)) - S(\rho)$. An adiabatically implemented channel can be seen as a special type of isothermal channel, where no heat is exchanged with the external thermal environment. By \defref{def:second-law}, $\Phi$ is consistent with the second law if and only if $\Delta S_{\h \to \kk} \geqslant 0$ for all $\rho$. 
\end{proof}
\

\begin{remark}\label{remark:state-prep-second-law}
All adiabatic state preparations are consistent with the second law. A state preparation is a channel that sends a trivial system $\co^1 \simeq \co |\Omega\>$ to a non-trivial system $\h$ as $\Upsilon_\rho : \lo(\co^1) \to \lo(\h), \proj{\Omega} \mapsto \Upsilon_\rho(\proj{\Omega}) = \rho$. But for any $\rho$ we have that $S(\rho) - S(\proj{\Omega}) = S(\rho) \geqslant 0$.     
\end{remark}

We now relate the above formulation to the more familiar unitary picture, in which the thermal bath is included explicitly as part of the quantum description. The basic idea is simple: an isothermal channel on a system can always be regarded as the reduced dynamics of a larger adiabatic unitary evolution, where the bath degrees of freedom are treated internally.

Let $\s$ and $\s'$ be systems with Hilbert spaces $\hs$ and $\h\sub{\s'}$, which may have different dimensions, and let $\bb$ be an auxiliary system playing the role of an internal thermal bath, with Hilbert space $\h\sub{\bb}$. Consider a channel $\ee : \lo(\hs \otimes \h\sub{\bb}) \to \lo(\h\sub{\s'} \otimes \h\sub{\bb})$. By the Stinespring dilation theorem, $\ee$ admits a unitary realization of the form
\begin{align}\label{eq:general-unitary-stinespring-channel}
    \ee(\bigcdot) = \tr\sub{\s+\cc}[U(\bigcdot \otimes \sigma)U^*] \, ,
\end{align}
where $\sigma$ is a state on $\h\sub{\s'} \otimes \h\sub{\cc}$, $\cc$ is an auxiliary system of sufficiently large dimension, and $U$ is a unitary operator acting on $\hs \otimes \h\sub{\bb} \otimes \h\sub{\s'} \otimes \h\sub{\cc}$.

We equip the total system with initial and final additive Hamiltonians
\begin{align}\label{eq:trivial-Hams}
    H_\mathrm{in} &= \hsys + H\sub{\bb} + \tilde H\sub{\s'} + H\sub{\cc} \, , \nonumber \\
    H_\mathrm{fin} &= \tilde H\sub{\s} + H\sub{\bb} + H\sub{\s'} + H\sub{\cc} \, ,
\end{align}
and choose
\begin{align*}
    \tilde H\sub{\s} = \lambda \one\sub{\s} \, , \qquad
    \tilde H\sub{\s'} = \lambda \one\sub{\s'} \, , \qquad
    H\sub{\cc} = \delta \one\sub{\cc} \, ,
\end{align*}
for arbitrary $\lambda,\delta \in \re$. With this choice, the final  and initial energies of the input  $\s$ and output $\s'$, respectively, cancel out, while the auxiliary system $\cc$ only serves as a bookkeeping device and does not affect the thermodynamic balance. 

Since unitary evolution exchanges no heat with an external environment, it is adiabatic, and the extracted work for an initial state $\varrho$ on $\hs \otimes \h\sub{\bb}$ is
\begin{align*}
    W_\mathrm{ext}
    = \tr[H_\mathrm{in}\,\varrho \otimes \sigma]
      - \tr[H_\mathrm{fin}\,U(\varrho \otimes \sigma)U^*] \, .
\end{align*}
With the above Hamiltonian choice, this reduces to
\begin{align*}
    W_\mathrm{ext}
    = \tr[(\hsys + H\sub{\bb})\varrho]
      - \tr[(H\sub{\s'} + H\sub{\bb})\ee(\varrho)] \, ,
\end{align*}
showing that $\ee$ is an adiabatically implemented channel on $\hs \otimes \h\sub{\bb}$.

Now let $\varrho = \rho \otimes \tau_\beta$, where $\rho$ is an arbitrary state on $\hs$ and $\tau_\beta = e^{-\beta H\sub{\bb}}/\tr[e^{-\beta H\sub{\bb}}]$ is the Gibbs state of the bath at inverse temperature $\beta$. Defining the reduced channel $\Phi : \lo(\hs) \to \lo(\h\sub{\s'})$ as
\begin{align*}
    \Phi(\bigcdot) \coloneq \tr\sub{\bb}[\ee(\bigcdot \otimes \tau_\beta)] \, ,
\end{align*}
we obtain an isothermally implemented channel, since its realization involves interaction with a single thermal bath $\bb$. For any $\rho$, the extracted work satisfies
\begin{align*}
    W_\mathrm{ext}
    \leqslant -\Delta F_{\s \to \s'}
      - \beta^{-1}\Delta S_{\s+\bb \to \s'+\bb} \, .
\end{align*}

Therefore, if the adiabatically implemented channel $\ee$ is consistent with the second law, so that $\Delta S_{\s+\bb \to \s'+\bb} \geqslant 0$ for all $\rho$, then the induced isothermal channel $\Phi$ is also consistent with the second law. In the special case $\ee = \Phi \otimes \idch\sub{\bb}$, corresponding to $U = V \otimes \one\sub{\bb}$, the reduced channel $\Phi$ itself is adiabatically implemented, and consistency with the second law is equivalent to the requirement that $\Phi$ does not decrease the entropy of any state on $\hs$.

Now we shall see how drawing the adiabatic boundary over a quantum instrument imposes constraints on its consistency with the second law:

\

\begin{lemma}\label{lemma:instrument-second-law}
 Let $\E \coloneq \{E_x : x\in \xx\}$ be a POVM on $\hs$, and $\ii \coloneq \{\ii_x : x \in \xx\}$ be an $\E$-instrument acting on $\hs$. If $\ii$ is implemented adiabatically, then it is consistent with the second law if and only if    
 \begin{align}\label{eq:instrument-second-law}
        \mathscr{H} (p^\E_\rho) \geqslant I_{\GO} (\ii, \rho) \qquad \forall \rho \in \s(\hs) \, .
    \end{align}
\end{lemma}
\begin{proof}
 Note that an instrument has both a quantum and a classical output, i.e., by \eq{eq:register-instrument-channel} and \eq{eq:avg-posterior-state},  for any input $\rho$ of the system, the expected output of system and classical register is 
\begin{align*}
   \Xi^\ii(\rho) = \sum_{x\in \xx} p^\E_\rho(x) \sigma^x\sub{\s} \otimes \proj{x}\sub{\rr} \, .  
\end{align*}
By \lemref{lemma:adiabatic-channel-second-law}, if  $\Xi^\ii$ is implemented adiabatically then it is consistent with the second law if and only if $S(\Xi^\ii(\rho)) \geqslant S(\rho)$ for all $\rho$. But since  $\{\proj{x}\}$ are mutually orthogonal, it holds that
\begin{align*}
 S(\Xi^\ii(\rho)) - S(\rho) =   \mathscr{H}(p^\E_\rho) - I_{\GO}(\ii, \rho) \, ,
\end{align*}
which is non-negative for all $\rho$ if and only if  \eq{eq:instrument-second-law} holds.    
\end{proof}

It is immediately clear that not all instruments can be adiabatically implemented in a way that is consistent with the second law. For example, consider $\hs = \co^d$ and a binary observable on $\hs$ given as $\E = \{E_+, E_-\}$ with the instrument $\ii_\pm(\bigcdot) = \tr[\bigcdot \, E_\pm] \proj{\psi_\pm}$ for an arbitrary pair of unit vectors $\ket{\psi_\pm}$. This instrument is quasicomplete, but not efficient, and we have that $I_{\GO}(\ii,\rho) = S(\rho)$ for all $\rho$.  Now note that $\mathscr{H}(p^\E_\rho) \leqslant \log(2)$ for any state $\rho$, whereas if we choose $\rho$ as the complete mixture then $I_{\GO}(\ii,\rho) = \log(d)$. Therefore, for $d > 2$ such an instrument violates \eq{eq:instrument-second-law}. 

The above notwithstanding, there are classes of instruments that are guaranteed to be consistent with the second law, when implemented adiabatically.  

\

\begin{corollary}\label{corollary:obj-bistochastic-efficient}
Each of the  following is a sufficient condition for an adiabatically implemented $\E$-instrument $\ii$ to be consistent with the second law, i.e., to  satisfy \eq{eq:instrument-second-law}:
\begin{enumerate}[(i)]
    \item $\ii_\xx$ is bistochastic.
    \item $\ii_x = \Phi_x \circ \ii^L_x$, where $\ii^L_x(\bigcdot) \coloneq \sqrt{E_x} \, \bigcdot \, \sqrt{E_x}$ is the generalized L\"uders instrument for $\E$ and $\Phi_x$ are arbitrary bistochastic channels.
    \item  $\ii$ is efficient.
\end{enumerate}
\end{corollary}

\begin{proof}

We first prove $(i)$.  If $\ii_\xx$ is bistochastic, then 
 \begin{align*}
 I_{\GO}(\ii , \rho) &= S(\rho) - \sum_{x\in \xx} p^\E_\rho(x) S(\sigma^x) \\
 & \leqslant S(\ii_\xx(\rho) ) - \sum_{x\in \xx} p^\E_\rho(x) S(\sigma^x) \\
 & = \chi(\{ p^\E_\rho(x), \sigma^x\}) \\
 &\leqslant \mathscr{H}(p^\E_\rho)\, , 
 \end{align*}
 where $\chi(\{ p^\E_\rho(x), \sigma^x\})$ is the Holevo $\chi$-quantity of the ensemble $\{ p^\E_\rho(x), \sigma^x\}$, which is known to be upper bounded by $\mathscr{H}(p^\E_\rho)$.

Now we prove $(ii)$. Note that $\ii_x^*(\oneapp) = \sqrt{E_x} \Phi_x^*(\oneapp)\sqrt{E_x} = \sqrt{E_x} \oneapp \sqrt{E_x} = E_x$ and so $\ii$ here is compatible with $\E$. Moreover,  $\ii_\xx(\oneapp) = \sum_x \Phi_x(E_x)$ which, even though $\Phi_x$ are bistochastic, in general does not equal $\oneapp$ unless $\Phi_x = \Phi_y$ for all $x,y$, and so in general $\ii_\xx$ is not bistochastic.  Denote $\tilde \sigma^x \coloneq \ii^L_x(\rho) / p^\E_\rho(x)$, and note that since $\Phi_x$ are bistochastic then $ S(\sigma^x) \geqslant S(\tilde \sigma^x) $ for all $x$ and prior states $\rho$. It follows that 
\begin{align*}
 I_{\GO}(\ii , \rho) &= S(\rho) - \sum_{x\in \xx} p^\E_\rho(x) S(\sigma^x)  \leqslant S(\rho ) - \sum_{x\in \xx} p^\E_\rho(x) S(\tilde \sigma^x) =  I_{\GO}(\ii^L , \rho)\, .
 \end{align*}
But, since $\ii^L_\xx$ is bistochastic then by the same arguments as in $(i)$ it holds that $I_{\GO}(\ii^L , \rho) \leqslant \mathscr{H}(p^\E_\rho)$, and so   $I_{\GO}(\ii , \rho) \leqslant \mathscr{H}(p^\E_\rho)$ for all $\rho$.   
 
Now we prove $(iii)$. This follows from $(ii)$ and noting that if $\Phi_x( \bigcdot) = U_x \, \bigcdot \, U_x^*$ are unitary channels then $\ii$ is efficient.

\end{proof}

Now we shall turn to the case where the adiabatic boundary is drawn not over the instrument on the measured system, but rather over the measurement process implementing it: 

\

\begin{definition}\label{def:measurement-process-second-law}
An adiabatically implemented measurement process is said to be \emph{consistent with the second law} whenever each of its constituent elements is consistent with the second law. Concretely, let  $\mm\coloneq  (\ha, \xi, \ee , \jj)$ be an adiabatically implemented measurement process for $\hs$. $\mm$ is consistent with the second law if and only if the following both hold: 
\begin{enumerate}[(i)]
    \item The premeasurement channel $\ee$ is  bistochastic.  

    \item The $\Z$-compatible objectification instrument $\jj$ satisfies   
    \begin{align}\label{eq:objectification-second-law}
        \mathscr{H} (p^\Z_\varrho) \geqslant I_{\GO} (\jj, \varrho) \qquad \forall \varrho \in \s(\ha) \, ,
    \end{align}
    where $p^\Z_\varrho(x) \coloneq \tr[Z_x \varrho] = \tr[\jj_x(\varrho)]$,  while the Shannon entropy $\mathscr{H}(p^\Z_\varrho)$ and the information gain $I_{\GO} (\jj, \varrho)$ are defined analogously to  \eq{eq:entropic-quantities}.
\end{enumerate}
\end{definition}

First, let us recall from \remref{remark:state-prep-second-law} that all adiabatic state preparations are consistent with the second law, so the constraints are imposed only on premeasurement and objectification. By \lemref{lemma:adiabatic-channel-second-law}, $\ee$ must not decrease the entropy of any state. Since $\ee$ has the same input and output,  if  $\ee$ is bistochastic, then it will satisfy this condition. But since the maximally mixed  state has the strictly highest possible entropy of any state in the system, if $\ee$ is not bistochastic then it will decrease the entropy of this state~\cite{Purves-2020}. Therefore, on the one hand we have $(i)$. On the other hand, $(ii)$ follows directly from \lemref{lemma:instrument-second-law}.

\

\begin{lemma}\label{lemma:augmented-instrument-second-law}
The second-law consistency of an adiabatically implemented measurement process is inherited by the instrument it induces on the joint system of system and apparatus. Concretely, let $\mm \coloneq (\ha, \xi, \ee, \jj)$ be an adiabatically implemented measurement process for $\hs$. If $\mm$ is consistent with the second law, then the corresponding global instrument acting on $\hs \otimes \ha$,
\begin{align*}
   (\idchsys \otimes \jj)\circ \ee \coloneq \{(\idchsys \otimes \jj_x)\circ \ee: x \in \xx\} \, ,
\end{align*}
which is compatible with the Heisenberg-evolved pointer observable $\Z^\tau \coloneq \{Z_x^\tau \equiv \ee^*(\onesys \otimes Z_x) : x \in \xx\}$, is itself consistent with the second law. That is, it satisfies
\begin{align}\label{eq:augmented-instrument-second-law}
    \mathscr{H}(p^{\Z^\tau}_\varrho) \geqslant I_{\GO}\bigl((\idchsys \otimes \jj)\circ \ee, \varrho\sub{\s\aa}\bigr)
    \qquad \forall \, \varrho\sub{\s\aa} \in \s(\hs \otimes \ha) \, .
\end{align}
\end{lemma}

\begin{proof}
Let us denote $\idchsys \otimes \jj \coloneq \{\idchsys \otimes \jj_x : x \in \xx\}$ as the extension of $\jj$  in $\hs \otimes \ha$. For any state $\varrho\sub{\s \aa}$ on $\hs \otimes \ha$, denote $\varrho\sub{\aa}\coloneq \trs[\varrho\sub{\s\aa}]$ as the reduced state in $\ha$, and by \eq{eq:register-instrument-channel} and \eq{eq:posterior-states} define
\begin{align}\label{eq:gamma-sak}
  \gamma\sub{\s \aa \rr} &\coloneq \idchsys \otimes \Xi^\jj(\varrho\sub{\s\aa}) =  \sum_{x \in \xx} \idchsys \otimes \jj_x (\varrho\sub{\s \aa}) \otimes |x\>\!\<x| = \sum_{x \in \xx} p^\Z_{\varrho}(x) \,  \gamma\sub{\s \aa}^x \otimes |x\>\!\<x| \, , 
\end{align}
where we note that $p^\Z_{\varrho}(x) = \tr[\jj_x(\varrho\sub{\aa})] = \tr[\idchsys \otimes \jj_x (\varrho\sub{\s \aa})]$,  that $\gamma\sub{\aa}^x \coloneq \trs[\gamma\sub{\s \aa}^x] = \jj_x(\varrho\sub{\aa}) / p^\Z_{\varrho}(x)$, and that 
\begin{align}\label{eq:gamma-ak}
\gamma\sub{\aa\rr} \coloneq \trs[\gamma\sub{\s\aa\rr}] = \Xi^\jj(\varrho\sub{\aa}) =  \sum_{x \in \xx} p^\Z_{\varrho}(x) \,  \gamma\sub{\aa}^x \otimes |x\>\!\<x| \, .
\end{align}
Then for any state $\varrho\sub{\s \aa}$, it holds that 
\begin{align*}
    I_{\GO}(\idchsys \otimes \jj, \varrho\sub{\s \aa} ) - I_{\GO}(\jj, \varrho\sub{\aa}) &= \bigg[S(\varrho\sub{\s \aa}) - S(\varrho\sub{\aa}) \bigg] -  \sum_{x \in \xx} p^\Z_{\varrho}(x) \bigg[ S(\gamma\sub{\s \aa} ^x) - S(\gamma\sub{\aa} ^x) \bigg] \\
    & = \bigg[S(\varrho\sub{\s \aa}) - S(\varrho\sub{\aa}) \bigg] - \bigg[ S(\gamma\sub{\s \aa \rr}) - S(\gamma\sub{\aa \rr}) \bigg]  \\
    & = S(\s|\aa)_\varrho - S(\s|\aa \rr)_\gamma \\
    & \leqslant 0 \, .
\end{align*}
Here, the second line follows from \eq{eq:gamma-sak} and \eq{eq:gamma-ak} and the fact that $\{\proj{x}\}$ are mutually orthogonal, while  the final line follows from the data processing inequality for coherent information and the fact that $\gamma\sub{\s\aa\rr}$ is obtained from $\varrho\sub{\s\aa}$ by a local channel  $ \Xi^\jj : \lo(\ha) \to \lo(\ha \otimes \hr)$ \cite[Theorem 11.9.3]{Wilde_2017}. 

If $\jj$ satisfies \eq{eq:objectification-second-law}, it follows that 
\begin{align*}
I_{\GO}(\idchsys \otimes \jj, \ee(\varrho\sub{\s\aa}) ) & \leqslant  I_{\GO}( \jj, \trs[\ee(\varrho\sub{\s\aa})] ) \leqslant \mathscr{H}(p^\Z_{\trs[\ee(\varrho)]}) = \mathscr{H}(p^{\Z^\tau}_{\varrho}) \,     
\end{align*}
holds for all $\varrho\sub{\s\aa}$, where the final equality follows from 
\begin{align*}
p^\Z_{\trs[\ee(\varrho)]}(x) = \tr[Z_x \trs[\ee(\varrho\sub{\s\aa})]] = \tr[\ee^*(\onesys \otimes Z_x) \varrho\sub{\s\aa}] = p^{\Z^\tau}_{\varrho}(x) \, .    
\end{align*}
Now define  $\sigma^x\sub{\s\aa} \coloneq  [(\idchsys \otimes \jj_x) \circ \ee](\varrho\sub{\s\aa})/ p^{\Z^\tau}_{\varrho}(x) $. If $\ee$ is bistochastic then for all $\varrho\sub{\s\aa}$  it holds that
\begin{align*}
 I_{\GO}(\idchsys \otimes \jj, \ee(\varrho\sub{\s\aa}) ) &= S(\ee(\varrho\sub{\s\aa})) - \sum_{x \in \xx} p^{\Z^\tau}_{\varrho}(x) S(\sigma^x\sub{\s \aa}) \\
 &\geqslant S(\varrho\sub{\s\aa})  - \sum_{x \in \xx} p^{\Z^\tau}_{\varrho}(x) S(\sigma^x\sub{\s \aa}) \\
 & =  I_{\GO}((\idchsys \otimes \jj) \circ \ee, \varrho\sub{\s\aa} ) \, .
\end{align*}
 Consequently, \eq{eq:augmented-instrument-second-law} holds.
\end{proof}

While not all adiabatically implemented instruments are  consistent  with the second law, all instruments admit an adiabatically implemented measurement process that is consistent with the second law:

\

\begin{lemma}
    Every instrument $\ii$ acting on $\hs$ admits an adiabatically implemented measurement process $\mm \coloneq (\ha, \xi, \ee, \jj)$ that is consistent with the second law. 
\end{lemma}
\begin{proof}
By the von Neumann--Naimark--Stinespring--Kraus--Ozawa dilation theorem~\cite{Ozawa1984},  any instrument can be realised by a ``pure'' process where  $\xi$ is pure, $\ee$ is  unitary (and hence bistochastic), and $\Z$ is projection valued. Such a state preparation and premeasurement satisfy the conditions laid out in \defref{def:measurement-process-second-law}, so now we just need to turn to the objectification instrument.   Recall that the choice of the particular $\Z$-compatible instrument $\jj$ acting on the apparatus is irrelevant for the system instrument $\ii$. Thus, without loss of generality, we may always assume that the measurement of $\Z$ happens via a generalized L\"uders instrument $\jj^L_x(\bigcdot ) \coloneq \sqrt{ Z_x}  \, \bigcdot  \, \sqrt{ Z_x}$, which gives rise to a bistochastic channel $\jj^L_\xx$. As a consequence, by \corref{corollary:obj-bistochastic-efficient} such a measurement process is consistent with the second law.    
\end{proof}

Indeed, consider an $\E$-instrument $\ii$ such that $I_{\GO}(\ii, \rho) > \mathscr{H}(p^\E_\rho)$ for some state $\rho$. Now let $\mm = (\ha, \xi, \ee, \jj^L)$ be a pure measurement process for $\ii$ as described above. In such a case, it holds that the probability distribution obtained by measuring the Heisenberg-evolved pointer observable $\Z^\tau$ on the state $\rho \otimes \xi$ of system and apparatus,  $p^{\Z^\tau}_{\rho \otimes \xi} $, equals the probability distribution obtained by measuring $\E$ in the system state $\rho$, $p^\E_\rho$, which follows from the fact that  
\begin{align*}
    p^{\Z^\tau}_{\rho \otimes \xi}(x) =  \tr[\ee^*(\onesys \otimes Z_x ) \rho \otimes \xi] = \tr[\Gamma_\xi \circ \ee^*(\onesys \otimes Z_x ) \rho ]  = \tr[E_x \, \rho] = p^\E_\rho(x) \, .
\end{align*}
But since $\mm$ is consistent with the second law, by  \lemref{lemma:augmented-instrument-second-law} it will hold that 
\begin{align*}
    \mathscr{H}(p^\E_\rho) \geqslant I_{\GO} ((\idchsys \otimes \jj^L) \circ \ee, \rho \otimes \xi)
\end{align*}
for any state $\rho$.   In other words, the second law does not impose a hard constraint on \textit{what} instruments can be implemented, only on \textit{how} they are implemented \cite{Minagawa2023a}.

\subsection{Consistency with the third law}

The third law of thermodynamics, or Nernst's unattainability principle, states that it is impossible for any deterministic adiabatic process to cool a system to absolute zero temperature with finite resources \cite{callen1991thermodynamics}. In the regime of finite-dimensional quantum systems, the third law is interpreted as implying the unattainability of pure states \cite{Schulman2005, Allahverdyan2011a, Freitas2018a, Taranto2021}. We now reiterate the operational definition for the third law presented in Ref. \cite{Mohammady2022a}.

\


\begin{definition}[Operational formulation of the third law]\label{def:third-law}
Let $\h$ and $\kk$ be finite quantum systems. We take the third law to assert that a deterministic adiabatic process between finite-dimensional quantum systems is possible if and only if the corresponding channel $\Phi:\lo(\h)\to\lo(\kk)$ is strictly positive.
\end{definition}

The argument for the above is as follows. A finite-dimensional quantum system, with Hamiltonian $H$, that is in  thermal equilibrium with respect to inverse temperature $\beta$ is in the Gibbs state $\tau_\beta \coloneq e^{-\beta H}/ \tr[e^{-\beta H}]$. Such states are strictly positive for non-vanishing temperatures, i.e.,  whenever $\beta < \infty$. On the other hand, if the Hamiltonian is non-trivial, i.e., $H \not\propto \one$, then at zero temperature ($\beta = \infty$) the Gibbs state is  rank-deficient. Since a deterministic quantum process is a channel, then the third law for finite-dimensional quantum systems can  be construed as the statement that for any adiabatically realisable channel $\Phi$, it must hold that  $\beta < \infty \implies \Phi(\tau_\beta) > \zero$. Since the existence of a strictly positive state in the image of a positive linear map is equivalent to the strict positivity of such a map, then an adiabatically implemented channel is consistent with the third law if and only if it is strictly positive \cite{VomEnde2022a}. 

\

\begin{remark}\label{remark:third-law-state}
An adiabatic state preparation $\rho$ on $\h$ is consistent with the third law if and only if $\rho$ is strictly positive. Recall that a state preparation  is given by the channel $\Upsilon_\rho : \lo(\co^1) \to \lo(\h), \proj{\Omega} \mapsto \Upsilon_\rho(\proj{\Omega}) = \rho$. $\Upsilon_\rho$ is consistent with the third law if and only if it is strictly positive. But $\proj{\Omega}$ is strictly positive in $\co^1$, and so $\Upsilon_\rho$ is strictly positive if and only if $\rho$ is strictly positive.    
\end{remark}

Contrast \defref{def:third-law} to \defref{def:second-law}. While the second law did not limit its applicability to quantum systems alone, the third law does. That is, the third law does not pertain to the classical register, which is a macroscopic device.  This plays a crucial role in what is to follow: 

\

\begin{lemma}\label{lemma:instrument-third-law}
 Let  $\ii \coloneq \{\ii_x : x \in \xx\}$ be an instrument acting on $\hs$. If $\ii$ is implemented adiabatically, then it is consistent with the third law if and only if $\ii_x$ is strictly positive for some $x\in \xx$.    
\end{lemma}
\begin{proof}
Recall that the outcomes of $\ii$ are stored in a classical register, which is a macroscopic device and so is not itself subject to the third law construed above for finite-dimensional quantum systems. However, the third law allows us to act on $\hs$ by strictly positive channels \emph{conditional} on the measurement outcomes. That is, after the application of $\Xi^\ii : \lo(\hs) \to \lo(\hs \otimes \hr)$ defined in \eq{eq:register-instrument-channel}, we may apply the  quantum-classical feedback channel $\Lambda(\bigcdot\sub{\s} \otimes \bigcdot\sub{\rr}) \coloneq \sum_x \Phi_x(\bigcdot\sub{\s}) \otimes \proj{x} \bigcdot\sub{\rr} \proj{x}$ on $\hs \otimes \hr$ where, for each $x$,  $\Phi_x$ are strictly positive channels acting on $\hs$.   The effective channel acting on $\hs$, given by the action of the instrument $\ii$ followed by feedback, is thus 
\begin{align*}
   \Theta^\ii_{\{\Phi_x\}}(\bigcdot) \coloneq \tr\sub{\rr} \circ \Lambda \circ \Xi^\ii (\bigcdot) =   \sum_{x \in \xx} \Phi_x \circ \ii_x(\bigcdot) \, .
\end{align*}
Therefore, by \defref{def:third-law} $\ii$ is consistent with the third law if and only if $\Theta^\ii_{\{\Phi_x\}}$ is strictly positive for all possible families $\{\Phi_x\}$ of strictly positive channels. Now we shall show that $\ii_x$ being strictly positive for some $x$ is both necessary and sufficient for this condition.   Sufficiency  follows from the fact that the composition of strictly positive operations is strictly positive, and a sum of positive maps is strictly positive if one of the maps is so.  Necessity follows from the fact that if none of $\ii_x$ are strictly positive, then they must each map the complete mixture $\one\sub{\s}/\dim(\hs)$ to a rank-deficient state. In such a case, unitary channels $\Phi_x$ may be chosen so that they map the support of each such conditional output state to the same proper subspace of $\hs$, implying that $\Theta^\ii_{\{\Phi_x\}}$ maps the complete mixture to a rank-deficient state.    
\end{proof}

\

\begin{remark}\label{remark:all-observables-sp-instrument}
    Note that restricting one of the operations $\ii_x$ to be strictly positive  does not in any way limit the  observable $\E$ that can be measured by $\ii$, since every effect $E_x$ admits  strictly positive $E_x$-compatible operations. Indeed, let $\E$ be an arbitary POVM and consider the $\E$-compatible instrument $\ii_x(\bigcdot) \coloneq \tr[ \bigcdot \, E_x] \sigma_x$ such that $\sigma_x > \zero$ for all $x$. In such a case, $\ii_x$ are strictly positive for all $x$ and, moreover, for all possible families of strictly positive channels $\{\Phi_x\}$ the channel $\Theta^\ii_{\{\Phi_x\}}(\bigcdot) = \sum_x \tr[\bigcdot \, E_x] \Phi_x(\sigma_x)$ maps all inputs to a strictly positive output.
\end{remark}

\

Now let us turn to the case where we draw the adiabatic boundary over the measurement process, and not the system instrument itself. By \defref{def:third-law}, \remref{remark:third-law-state}, and \lemref{lemma:instrument-third-law} we obtain the following:

\

\begin{definition}\label{def:measurement-process-third-law}
An adiabatically implemented measurement process is said to be consistent with the third  law whenever each of its constituent elements is consistent with the third law. Concretely, let  $\mm\coloneq  (\ha, \xi, \ee , \jj)$ be an adiabatically implemented measurement process for $\hs$. $\mm$ is consistent with the third law if and only if the following all hold: 
\begin{enumerate}[(i)]
    \item The apparatus state preparation $\xi$ is strictly positive.

    \item The premeasurement channel $\ee$ is  strictly positive.  

    \item For at least one outcome $x$, the operation $\jj_x$ of the $\Z$-compatible objectification instrument $\jj$ is strictly positive. 
\end{enumerate}
\end{definition}

\

\begin{lemma}\label{lemma:instrument-measurement-process-third-law}
Let $\ii \coloneq \{\ii_x: x \in \xx\}$ be an instrument acting on $\hs$. $\ii$ admits an adiabatically implemented measurement process $\mm\coloneq  (\ha, \xi, \ee , \jj)$ that  is consistent with the third law if and only if $\ii_x$   is strictly positive for all $x$. 
\end{lemma}
\begin{proof}
The necessity was shown in \cite[Lemma D.1]{Mohammady2022a}, and the sufficiency was later shown in \cite[Proposition E.1]{Shahbeigi2025} . We shall provide a simple proof for necessity here. Let us note that we may equivalently write \eq{eq:system-instrument} as
\begin{align*}
    \ii_x(\bigcdot) = \tra \circ (\idchsys \otimes \jj_x)\circ \ee \circ \Upsilon_\xi(\bigcdot)
\end{align*}
where $\Upsilon_\xi : \lo(\hs) \to \lo(\hs \otimes \ha), \rho \mapsto \rho \otimes \xi$ is the preparation map,  which is strictly positive if $\xi > \zero$. The partial trace is also strictly positive.  Now, recall that $\ii_x$ is uniquely determined by the effect $Z_x$ that is measured by $\jj_x$, and not on $\jj_x$ itself. Since every effect admits a strictly positive operation (\remref{remark:all-observables-sp-instrument}), without loss of generality we may assume that $\jj_x$ are strictly positive for all $x$. If $\mm$ is consistent with the third law, by \defref{def:third-law}  $\xi$ is a strictly positive state, and $\ee$ is a strictly positive channel. In such a case $\ii_x$ can be seen as being realised by a composition of strictly positive maps, and so it must be strictly positive, for all $x$.  
\end{proof}

\section{Consistency with both laws: efficient instruments and measurement processes}

Now we shall consider the consistency of adiabatically implemented instruments and measurement processes  with both the second and third laws, in particular highlighting the impact of where we draw the adiabatic boundary on the realisability of  efficient measurements. First, let us recall a useful result:

\

\begin{lemma}\label{lemma:purity-preserving-strictly-positive}
    Let $\Phi$ be an operation acting on $\h$. Assume that $\Phi$ is purity-preserving and strictly positive, in the sense that $A > \zero$ implies $\Phi(A) > \zero$. Then
    \[
        \Phi(\bigcdot) = U \sqrt{E} \, \bigcdot \, \sqrt{E} U^*
    \]
    with $E$ a strictly positive effect and $U$ a unitary operator.
\end{lemma}


\begin{proof}
Every operation is compatible with an effect $E$. By Theorem 3.1 of Ref. \cite{Davies1976}, a purity-preserving operation is either 
(i) $\Phi(\bigcdot ) = K \, \bigcdot \, K^*$ for some $K \in \lo(\h)$ such that $K^*K = E$, or $(ii)$ $\Phi(\bigcdot ) = \tr[ E \, \bigcdot ] |\phi\>\!\<\phi|$ with $|\phi\>$ a unit vector in $\h$. Option $(ii)$ is evidently not strictly positive, so we are left with option $(i)$. By the polar decomposition, it holds that $K= U \sqrt{E}$. Now note that $\Phi$ is strictly positive if and only if $\Phi(\one)$ is strictly positive \cite{VomEnde2022a}. Since $\Phi(\one) = U E U^*$, it follows that $E$ must be strictly positive.
\end{proof}

Now we are ready to prove our first main result:

\

\begin{theorem}\label{thm:efficient-instrument-both-laws}
    Let $\E\coloneq \{E_x : x\in \xx\}$ be an observable on $\hs$. The following hold:
    \begin{enumerate}[(i)]
        \item  An adiabatically implemented $\E$-instrument $\ii \coloneq \{\ii_x : x\in \xx\}$  acting on $\hs$  is consistent with both the second and third laws if and only if 
        \begin{align}\label{eq:second-law-thm}
            \mathscr{H}(p^\E_\rho) \geqslant I_{\GO} (\ii, \rho) \qquad \forall \rho \in \s(\hs)
        \end{align}
        and $\ii_x$ is strictly positive for some $x$. 

        \item $\E$ admits an adiabatically implemented \emph{efficient} instrument $\ii$ that is consistent with both the second and third laws if and only if $E_x > \zero$ for some $x$.  
        
    \end{enumerate}

\end{theorem}
\begin{proof}
Item $(i)$ follows directly from the conjunction of \lemref{lemma:instrument-second-law} and \lemref{lemma:instrument-third-law}. So now we shall prove $(ii)$.  Recall from \corref{corollary:obj-bistochastic-efficient} that all adiabatically implemented efficient instruments satisfy \eq{eq:second-law-thm}, i.e., they are consistent with the second law, so we need only consider the constraint imposed by the third law: an adiabatically implemented instrument is consistent with the third law if and only if $\ii_x$ is strictly positive for some $x$. So, consider such an operation $\ii_x$. If $\ii$ is an efficient instrument, then this operation is purity preserving and strictly positive.  Then  by \lemref{lemma:purity-preserving-strictly-positive} it must be compatible with a strictly positive effect. On the other hand, for any strictly positive effect $E_x$, the corresponding operation $\ii_x(\bigcdot) = U_x \sqrt{E_x} \bigcdot \sqrt{E_x} U^*$ is completely purity preserving and strictly positive, and so it can belong in the range of an adiabatically implemented efficient instrument that is consistent with the third law.  
\end{proof}

\

\begin{corollary}[First No-Go Result]\label{cor:projective-no-efficient}
Let $\E$ be a non-trivial projective observable. Then $\E$ does not admit a thermodynamically consistent purity-preserving instrument. This follows from the fact that all effects of a non-trivial projective observable are projections with both 0 and 1 in their spectrum,  and hence none is strictly positive.
\end{corollary}

Now let us  move from instruments to measurement processes, and consider the case where the \textit{process} is consistent with both the second and third laws. By the conjunction of  \defref{def:measurement-process-second-law} and \defref{def:measurement-process-third-law}, and recalling that bistochastic channels are also strictly positive,   we obtain the following:

\




\begin{definition}[Consistency of a measurement process with both laws]\label{def:measurement-process-both-law}
Let $\mm\coloneq(\ha,\xi,\ee,\jj)$ be an adiabatically implemented measurement process for $\hs$. We say that $\mm$ is consistent with both the second and third laws if and only if it satisfies the second-law criterion and the third-law criterion simultaneously. Equivalently, the following hold:
\begin{enumerate}[(i)]
    \item The apparatus state $\xi$ is strictly positive.

    \item The premeasurement channel $\ee$ is bistochastic.

    \item The $\Z$-compatible objectification instrument $\jj\coloneq\{\jj_x:x\in\xx\}$ satisfies
    \begin{align*}
        \mathscr{H}(p^\Z_\varrho)\geqslant I_{\GO}(\jj,\varrho)
        \qquad \forall\,\varrho\in\s(\ha),
    \end{align*}
    and there exists at least one outcome $x\in\xx$ such that $\jj_x$ is strictly positive.
\end{enumerate}
\end{definition}

\

 \begin{remark}\label{remark:obj-instrument-both-laws}
  Item $(iii)$ of the above does not impose any constraints on the pointer observable $\Z \coloneq \{Z_x : x \in \xx\}$. Let $d= \dim(\ha)$. For an arbitrary $\Z$, choose $\jj_{x_0}(\bigcdot) = \tr[\bigcdot Z_{x_0}] \oneapp/ d$ for some $x_0$  and $\jj_x(\bigcdot) = \sqrt{Z_x} \bigcdot \sqrt{Z_x}$ for all $x \ne x_0$. That is, all operations correspond to those of the L\"uders instrument $\jj^L$, except for the operation $\jj_{x_0}$ which maps all inputs to the complete mixture. This instrument has at least one strictly positive operation and so is consistent with the third law. Furthermore, defining $\gamma^x \coloneq \jj_x(\varrho)/ p^\Z_\varrho(x)$,  we have that 
  \begin{align*}
      I_{\GO}(\jj, \varrho) &= S(\varrho) -  p^\Z_\varrho(x_0) \ln(d) -\sum_{x\ne x_0} p^\Z_\varrho(x) S(\gamma^x)   \leqslant I_{\GO} (\jj^L, \varrho)  \leqslant \mathscr{H}(p^\Z_\varrho) \, 
  \end{align*}
 for all $\varrho$, and so the instrument is also consistent with the second law.  The first inequality  follows from the fact that the entropy of $\sqrt{Z_{x_0}} \varrho \sqrt{Z_{x_0}} /p^\Z_\varrho(x_0)$ is not larger  than $\ln(d)$ for any $\varrho$, and the second inequality follows from \corref{corollary:obj-bistochastic-efficient} .
 \end{remark}

We are now ready to present our second main result, which is that a thermodynamically consistent adiabatic measurement process  can \textit{never} implement a quasicomplete instrument, let alone an efficient one.

\

\begin{theorem}[Second No-Go Result]\label{thm:bistochastic-third-law-purity-preserving-nogo}
Let $\mm\coloneq (\ha, \xi, \ee, \jj)$ be  an adiabatically implemented measurement process for an  instrument $\ii \coloneq  \{\ii_x : x\in \xx\}$ acting on $\hs$, and assume that $\mm$ is consistent with both the second and third laws. It follows that any operation $\ii_x$,  that is compatible with a non-trivial effect $E_x$,  cannot be purity-preserving. 
\end{theorem}

\

It should be clear that an operation compatible with a trivial effect provides no information about the system being measured, since its outcome probability is independent of the input state, and it therefore does not contribute to the informational content of the measurement. In fact, since by \lemref{lemma:instrument-measurement-process-third-law} the third law requires that the operations of the instrument be strictly positive, \lemref{lemma:purity-preserving-strictly-positive} implies that if an operation $\ii_x$ is compatible with a trivial effect $E_x = p(x) \onesys$, then it is purity preserving only if it is proportional to a unitary channel, i.e, $\ii_x(\bigcdot) = p(x) U_x \bigcdot U_x^*$.

Before we prove the above theorem, we shall first prove the following useful lemma:

\

\begin{lemma}\label{lemma:purity-strictpositive-channel-identity}
Let  $\mm\coloneq (\ha, \xi, \ee, \jj)$ be  a measurement process for an instrument $\ii$  acting on $\hs$. Assume that $\xi$ is strictly positive. For any outcome $x$ such that $\ii_x$ is purity-preserving, it holds that
\begin{align*}
    \ee^*(\bigcdot  \otimes  Z_x) = \ii_x^*(\bigcdot ) \otimes \oneapp .
\end{align*}
We recall that $\Z$ is the observable compatible with the apparatus instrument $\jj$.
\end{lemma}

\begin{proof}
If  $\xi$ is strictly positive, then   for an arbitrary unit vector $\ket{\phi} \in \ha$, there exists a $0< \lambda < 1$ such that $\xi > \lambda |\phi\>\!\<\phi|$. Defining the state $\sigma\coloneq  (\xi - \lambda |\phi\>\!\<\phi|)/ (1 - \lambda)$, we may thus decompose $\xi$ as $\xi = \lambda |\phi\>\!\<\phi| + (1-\lambda) \sigma$. By \eqref{eq:system-instrument}, and  the fact that for any decomposition $\xi = \sum_i q_i \xi_i$ it holds that $\Gamma_\xi(\bigcdot ) =\sum_i q_i \Gamma_{\xi_i}(\bigcdot )$, we have that 
\begin{align*}
    \ii_x^*(\bigcdot ) &= \Gamma_\xi\circ\ee^*(\bigcdot  \otimes  Z_x) \nonumber \\
    & = \lambda \, \Gamma_{|\phi\>\!\<\phi|}\circ\ee^*(\bigcdot  \otimes  Z_x) + (1-\lambda) \Gamma_\sigma\circ\ee^*(\bigcdot  \otimes  Z_x) \\
    & = \lambda \, {\ii_x^\phi}^*(\bigcdot ) + (1-\lambda) {\ii_x^\sigma}^*(\bigcdot ),
\end{align*}
where we have defined ${\ii_x^\phi}^*(\bigcdot )\coloneq  \Gamma_{|\phi\>\!\<\phi|}\circ\ee^*(\bigcdot  \otimes  Z_x)$ and ${\ii_x^\sigma}^*(\bigcdot )\coloneq  \Gamma_{\sigma}\circ\ee^*(\bigcdot  \otimes  Z_x)$. By the trace duality it holds that $\ii_x(\bigcdot ) = \lambda \,  {\ii_x^\phi}(\bigcdot ) + (1-\lambda) {\ii_x^\sigma}(\bigcdot )$. Since $\ii_x$ is assumed to be purity-preserving, then for every pure state $\rho$ on $\hs$ it must hold that $\ii_x(\rho) = \ii_x^\phi(\rho) = \ii_x^\sigma(\rho)$, since if it were otherwise then $\ii_x(\rho)$ would be mixed. By linearity, it follows that $\ii_x(\bigcdot ) = \ii_x^\phi(\bigcdot )$ for all  unit vectors $|\phi\>$ in $\ha$. Therefore, 
\begin{align*}
    \ii_x^*(\bigcdot ) &=  \Gamma_{|\phi\>\!\<\phi|} \circ \ee^*(\bigcdot  \otimes  Z_x)
\end{align*}
must hold for arbitrary unit vectors $\ket{\phi} \in \ha$. Writing an arbitrary state $\varrho = \sum_i q_i |\phi_i\>\!\<\phi_i|$, it follows that  
\begin{align*}
    \ii_x^*(\bigcdot ) &= \sum_i q_i \Gamma_{|\phi_i\>\!\<\phi_i|} \circ \ee^*(\bigcdot  \otimes  Z_x) = \Gamma_\varrho \circ \ee^*(\bigcdot  \otimes  Z_x)
\end{align*}
must hold for any state $\varrho$ on $\ha$. As shown in  Lemma I.2 of Ref. \cite{Mohammady2022a},  for any $A\in \lo(\hs\otimes \ha)$ and $B\in \lo(\hs)$ such that  $B =\Gamma_\varrho(A)$ for any choice of  $\varrho$, it holds that  $A = B \otimes \oneapp$. This completes the proof.
\end{proof}

Now we may prove \thmref{thm:bistochastic-third-law-purity-preserving-nogo}.

\begin{proof}[Proof of \thmref{thm:bistochastic-third-law-purity-preserving-nogo}]
By \defref{def:measurement-process-both-law}, the apparatus state $\xi$ must be strictly positive, and the premeasurement channel $\ee$ must be bistochastic. On the other hand, we recall that the instrument $\ii$ acting on the system depends only on the pointer observable $\Z$ and not on the $\Z$-instrument $\jj$, and we recall that thermodynamic consistency does not limit $\Z$ in any way.  Assume that  for some $x$, the operation $\ii_x$ is purity-preserving. Since $\xi$ is strictly positive,  By \lemref{lemma:purity-strictpositive-channel-identity} it holds that  $\ee^*(\bigcdot  \otimes  Z_x)
= \ii_{x}^*(\bigcdot )\otimes \oneapp$. Since $\ee$ is bistochastic, then $\ee^*$ preserves the trace. Therefore, for every  state $\rho \in \s(\hs)$ it holds that 
\begin{align*}
\tr[\rho \otimes  Z_x] &= \tr[\ee^*(\rho \otimes  Z_x)] =  \tr[\ii_x^*(\rho) \otimes \oneapp]
\end{align*}
and so 
\begin{align*}
 \tr[\ii_x^*(\rho)] = \frac{\tr[  Z_x]}{\dim(\ha)} \qquad \forall \, \rho.
\end{align*}
Since the measurement process is consistent with the third law, then by \lemref{lemma:instrument-measurement-process-third-law}  $\ii_x$ is strictly positive. Since we assume that $\ii_x$ is purity-preserving, then by  \lemref{lemma:purity-preserving-strictly-positive} we may  write $\ii_x^*(\bigcdot ) = \sqrt{E_x} U_x^* \bigcdot  U_x \sqrt{E_x}$, with $U_x$ a unitary operator. It follows that
\begin{align*}
 \tr[U_x E_x U_x^* \rho ] = \frac{\tr[  Z_x]}{\dim(\ha)} \qquad \forall \, \rho.
\end{align*}
Note that $U_x E_x U_x^*$ is an effect, which is trivial if and only if $E_x$ is trivial.  Since the right hand side is independent of  $\rho$, it follows that $E_x$ must be trivial. 
\end{proof}

\thmref{thm:bistochastic-third-law-purity-preserving-nogo} immediately leads to the following corollary:

\

\begin{corollary}
  Let $\E$ be a non-trivial observable on $\hs$, and let $\ii$ be an $\E$-instrument acting on $\hs$. Assume that $\ii$ admits an adiabatically implemented measurement process  $\mm$ that is consistent with both the second and third laws.  The following hold:

  \begin{enumerate}[(i)]
        \item $\ii$ is not a purity-preserving instrument; in particular, it is not an efficient instrument.
      \item Every operation $\ii_x(\bigcdot )$ that is compatible with a non-trivial effect $E_x$  has a minimal Kraus representation with at least two Kraus operators.
      \item There exists some state $\rho$ such that the information gain $I_{\GO}(\ii, \rho)$ is strictly negative. 
       
  \end{enumerate}
\end{corollary}

\subsection{Approximately purity-preserving measurements}

While a thermodynamically consistent, adiabatically implemented measurement process can never realise a strictly purity-preserving instrument, it can nevertheless realise instruments that are arbitrarily close to being purity-preserving.

Let $\tilde \ii$ be an arbitrary instrument. By the standard unitary measurement model, there exists a measurement process $(\ha, \proj{0}, \ee, \jj)$ implementing $\tilde \ii$, where the apparatus is prepared in a pure state $\proj{0}$, the premeasurement channel $\ee$ is unitary, and the objectification instrument $\jj$ is compatible with a projection-valued pointer observable $\Z$. In such a realization, the adiabatic implementation of $\ee$ is thermodynamically consistent, and by \remref{remark:obj-instrument-both-laws} the instrument $\jj$ can always be chosen to be thermodynamically consistent when implemented adiabatically. The only thermodynamically inconsistent element of this process is therefore the adiabatic preparation of the pure apparatus state $\proj{0}$.

Consider instead the apparatus state
\[
\xi = (1-\epsilon)\proj{0} + \epsilon\,\omega ,
\]
where $\omega > \zero$ and $0 < \epsilon < 1$ can be arbitrarily small. Then $\xi$ is strictly positive, and the resulting measurement process $\mm \coloneq (\ha, \xi, \ee, \jj)$ is thermodynamically consistent. It realises the instrument
\[
\ii_x(\bigcdot) = (1-\epsilon)\tilde\ii_x(\bigcdot) + \epsilon\,\Phi_x(\bigcdot) ,
\]
where $\Phi$ is the instrument implemented by $(\ha, \omega, \ee, \jj)$.

Recalling that the trace norm is defined as $\| A \|_1 \coloneq \tr[\sqrt{A^*A}]$, for any state $\rho$ it holds that
\begin{align*}
    \| \tilde\ii_x(\rho) - \ii_x(\rho)\|_1
    &= \epsilon \| \tilde\ii_x(\rho) + \Phi_x(\rho) \|_1 \\
    &\leqslant \epsilon \| \tilde\ii_x(\rho) \|_1 + \epsilon \| \Phi_x(\rho) \|_1 \\
    &\leqslant 2\epsilon \, .
\end{align*}
Therefore, for any prior state $\rho$, the posterior state produced by the instrument $\ii$ is $\epsilon$-close to that produced by $\tilde\ii$. In particular, if $\tilde\ii$ is purity-preserving, then for every pure state $\rho$ the corresponding posterior state of $\ii$ is $\epsilon$-close to a pure state.

\section{Discussion}

We have shown that the level at which the adiabatic boundary is drawn during a quantum measurement---either at the level of the instrument acting on the system or at the level of the measurement process implementing that instrument---has decisive consequences for the thermodynamic realisability of measurements.

A striking feature of our analysis is that the second and third laws of thermodynamics play qualitatively different roles in the two scenarios. When the adiabatic boundary is drawn at the level of the instrument, consistency with the second law is highly restrictive: an instrument is thermodynamically admissible if and only if the Shannon entropy of the outcome probabilities exceeds the Groenewold--Lindblad--Ozawa information gain. By contrast, every instrument admits an adiabatically implemented measurement process that is consistent with the second law. In this sense, the second law constrains what instruments can be implemented adiabatically, but not what instruments can be realised through an adiabatic measurement process.

The situation is reversed for the third law. At the level of the instrument, thermodynamic consistency requires only that at least one operation be strictly positive. However, when the adiabatic boundary is drawn over the full measurement process, consistency with the third law becomes substantially stronger: an instrument admits a thermodynamically consistent adiabatic implementation if and only if all of its operations are strictly positive.

The conjunction of the second and third laws therefore leads to a sharp separation between the two notions of thermodynamic closure. While an observable admits an adiabatically implemented \emph{efficient} instrument that is consistent with both laws if and only if at least one of its effects is strictly positive, no non-trivial observable admits an adiabatically implemented measurement process that is consistent with both laws and realises a quasicomplete instrument, let alone an efficient one.

This result has a clear conceptual implication. If one insists that efficient measurements can be implemented without violating the laws of thermodynamics, then one must abandon the universal validity of the unitary interaction-based indirect measurement model. In other words, one must accept the existence of \textit{thermodynamically closed measurement processes that are not extendable into a larger unitary dynamics}.

Of course, we do not claim that thermodynamics invalidates the indirect measurement model as a mathematical representation of instruments. What our results show is narrower and stronger at the same time: the model cannot be used indiscriminately as a universal physical implementation principle once thermodynamic consistency is imposed at the level of the measurement process. This tension, between information-theoretic universality and thermodynamic admissibility, deserves further study and may bear on broader questions about reductionism and the status of enlarged unitary descriptions in quantum theory~\cite{DAriano2020a,cuffaro,Wallace2024,Chen2024}. We leave a more systematic investigation of this tension for future work.

\begin{acknowledgments}
The authors thank Fereshte Shahbeigi, Cyril Elouard,  and Maximilian P. E. Lock for insightful discussions. M.~H.~M. acknowledges funding provided by the IMPULZ program of the Slovak Academy of Sciences under the Agreement on the Provision of Funds No. IM-2023-79 (OPQUT), as well as from projects  VEGA 2/0164/25 (QUAS) and APVV-22-0570 (DeQHOST). F.~B. acknowledges support from MEXT Quantum Leap Flagship Program (MEXT QLEAP) Grant No.~JPMXS0120319794 and from JSPS KAKENHI, Grants No.~23K03230 and No.~26K00621.
\end{acknowledgments}



\bibliography{Projects-Thermal-Measurement.bib}

\end{document}